\documentclass[letterpaper]{article} 
\usepackage{aaai2026}  
\usepackage{times}  
\usepackage{helvet}  
\usepackage{courier}  
\usepackage[hyphens]{url}  
\usepackage{graphicx} 
\usepackage{multirow} 
\usepackage{tabularx} 
\usepackage{amsmath} 
\usepackage{amssymb} 
\usepackage{placeins} 
\usepackage{subcaption} 
\usepackage{algorithm}
\usepackage[noend]{algorithmic}
\usepackage{amsthm}
\theoremstyle{definition}
\newtheorem{definition}{Definition}
\theoremstyle{plain}
\newtheorem{theorem}{Theorem}
\newtheorem{lemma}{Lemma}
\newtheorem{corollary}{Corollary}
\newtheorem{proposition}{Proposition}
\usepackage{natbib}  
\usepackage{caption} 
\nocopyright

\title{Fixed-Haven Reservation for Online Multi-Agent Pickup and Delivery \\ in Dense Warehouses}
\author{
    Taisei Hirayama$^1$, Kohei Yoshida$^2$, Hiroki Sakaji$^1$, and Itsuki Noda$^1$
}
\affiliations{
    $^1$Hokkaido University, Sapporo 060-0808, Japan\\
    $^2$Toyota Industries Corporation, Aichi 474-8601, Japan\\
    hirayama.h77@gmail.com, \{sakaji, i.noda\}@ist.hokudai.ac.jp
}

\begin{document}

\maketitle

\begin{abstract}
Dense warehouses often contain single-lane aisles, dead ends, and tree-like guidepaths that leave little room for idle agents to wait without blocking others.
Existing Multi-Agent Pickup and Delivery (MAPD) guarantees for completing all finitely released tasks typically rely on extra waiting endpoints that planned paths can avoid, or on biconnected topology; these assumptions may fail in such layouts.
We study fixed-Haven reservation for online MAPD, where pickup-delivery tasks are released over time.
Each agent owns a fixed Safe Haven (Haven for short), usually its start cell, that only the owner may occupy and that other agents treat as blocked.
For finite task releases, we prove that this fixed-Haven contract completes all released tasks under Haven-Reachability and explicit planning/progress assumptions.
We implement the contract in SHARP, a Safe-Haven Retreat Planner that keeps every busy or retreating agent on a collision-free reserved route ending at its Haven.
We compare SHARP with representative TP and PIBT-family MAPD baselines: Token Passing (TP), Priority Inheritance with Backtracking (PIBT), and PIBT with Temporary Priority and Temporary Avoidance (PIBTTP-TA) for biconnected main areas with attached trees.
In the robustness sweep, SHARP is the only method with 100\% success on all tested configurations, at substantially higher centralized planning cost on tree-like layouts.
A TP-style fixed-home-return counterfactual with full-route validation also recovers robustness on tested tree-like layouts, suggesting that fixed return is a central robustness mechanism there.
A no-overwrite variant shows that disabling mid-retreat reassignment worsens service time (release-to-delivery latency) by 1.89$\times$ and makespan by 1.53$\times$ in the tested high-load tree condition.
\end{abstract}

\section{Introduction}

Multi-Agent Path Finding (MAPF) asks agents to reach goals without colliding~\citep{stern2019multi}.
Multi-Agent Pickup and Delivery (MAPD) extends this model to online pickup-delivery tasks: tasks are released over time, each task has a pickup and a delivery location, and an assigned agent must visit them in order~\citep{ma2017lifelong, ma2019lifelong}.
We focus on the finite task horizon setting, where only finitely many tasks are released.
We use \emph{complete} in this finite-release MAPD sense: every released task is delivered in finite time.

Warehouse robot fleets are a canonical MAPD setting~\citep{wurman2008coordinating}; there, MAPD is difficult not only because paths must avoid collisions, but also because agents that are waiting, returning, or newly assigned can block narrow passages.
This issue is acute in layouts with single-lane aisles, dead ends, or tree-like guidepaths, which arise in warehouse and storage/retrieval abstractions~\citep{azadeh2017robotized, roy2017multitier}.
In such graphs, a planner needs a place where every agent can eventually clear the task area without relying on wide detours.

Existing MAPD guarantees address this problem through structural assumptions.
Token Passing (TP) is complete on well-formed MAPD instances, where agents have enough separate waiting endpoints and paths between endpoints can avoid other endpoints~\citep{ma2017lifelong}.
Priority Inheritance with Backtracking (PIBT) guarantees reachability on biconnected graphs, and the Temporary Avoidance variant of PIBT with Temporary Priority (PIBTTP-TA) extends this idea to a biconnected main area with attached trees under additional restrictions~\citep{okumura2019priority, fujitani2022priority}.
These assumptions are valuable, but they do not cover all dense warehouse guidepath graphs.

SHARP, a Safe-Haven Retreat Planner, studies a complementary contract.
Each agent owns a fixed Safe Haven, normally its start cell; non-owners are not allowed to occupy it.
Whenever an agent is executing a task or returning after delivery, SHARP maintains a collision-free reservation from the agent's current location through any unfinished task waypoint and then back to its Haven.
A retreating agent may accept a new task before reaching its Haven, but only if the planner can replace the remaining retreat suffix with a newly validated pickup-delivery-Haven route.

The main contribution is a fixed-Haven contract and a finite-release completion guarantee under Haven-Reachability.
SHARP is one centralized graph-level implementation of this contract: it validates routes with Safe Interval Path Planning (SIPP), keeps owner-only Havens protected, and keeps committed tasks with their assigned agents.
We then use experiments to separate three effects: which layout classes the fixed-Haven contract covers against representative baselines, whether a TP-style policy with fixed-home return and full-route validation recovers robustness, and how much efficiency comes from allowing mid-retreat reassignment via suffix overwrite.
SHARP is a graph-level dispatch-and-routing component for lower-level controllers; perception, localization, and continuous control are out of scope.

\section{Related Work and Positioning}
\label{sec:related}

\subsection{MAPD variants and structural assumptions}

MAPD is \emph{online} when tasks are revealed over time, and \emph{offline} when the full task set is known in advance~\citep{ma2017lifelong, liu2019taskpathplanning}.
The standard MAPD task has two goal locations, pickup and delivery; multi-goal MAPD generalizes this to task sequences with more goals~\citep{xu2022multigoal}.
Other MAPD variants add operational constraints such as capacity and integrated task-assignment/path-planning decisions~\citep{chen2021integrated}.
We study standard online MAPD with finite task releases.

TP~\citep{ma2017lifelong} is complete on \emph{well-formed} MAPD instances.
In that definition, endpoints are task locations or separate resting/parking locations: task endpoints are pickup/delivery locations, non-task endpoints are separate waiting locations, there are enough non-task endpoints for the agents, and any two endpoints are connected by a path that avoids all other endpoints.
Related well-formed infrastructure assumptions also appear in online multi-robot trajectory planning~\citep{cap2015complete}.
This condition is elegant but restrictive when pickup and delivery locations lie on narrow corridors, because those locations become endpoints that paths must avoid.

PIBT~\citep{okumura2019priority} guarantees reachability on graphs where adjacent vertices lie on cycles, with biconnected graphs as the main undirected case.
\citet{fujitani2022priority} extend the idea to a biconnected main area with attached trees via PIBT with Temporary Priority (PIBTTP) and its Temporary Avoidance (TA) variant PIBTTP-TA.
The tree extension also comes with task-placement and motion restrictions inside trees.
These assumptions differ from well-formedness, but they are still graph-structure assumptions rather than reservation-invariant assumptions.

\subsection{Fallback, parking, and standby}

Several MAPD methods already use the intuition that agents need safe fallback locations.
\citet{liu2019taskpathplanning} reserve dummy paths to each agent's parking location in offline MAPD and note that a parking-avoidance variant of well-formedness is sufficient for their proof.
\citet{xu2022multigoal} use dummy endpoints in multi-goal MAPD and may reassign them from an endpoint pool.
\citet{yamauchi2022standby} choose standby nodes dynamically in maze-like few-endpoint environments, and home-location variants have also been studied in multi-item logistics~\citep{farinelli2020decentralized, contini2021coordination}.
These are relevant precedents, but SHARP's contract is different: it protects only fixed owner-only Havens chosen from agent starts, maintains an executable return suffix to the owner's Haven, and proves completion through quiescent states in which the task area clears.
More specifically, SHARP differs in what must be avoided structurally, where fallback cells come from, what invariant is maintained during execution, and how completeness is argued:
\begin{itemize}
\item only per-agent exclusive Havens, rather than all endpoints, are protected;
\item each Haven is fixed to an agent start cell, so no separate shared fallback pool is required;
\item every busy agent keeps an executable retreat reservation to its own Haven at all times; and
\item the proof uses convergence to a quiescent state in which the task area has cleared.
\end{itemize}
Table~\ref{tab:positioning} summarizes the closest explicit fallback contracts and completeness-style assumptions.
In the remainder, ``fixed-Haven framework'' refers to the broader guarantee framework, whereas ``SHARP'' refers to the specific centralized policy studied here.

\paragraph{TP-family reference points.}
We use TP-family names carefully.
The original \emph{TP} baseline~\citep{ma2017lifelong} uses endpoint-resting semantics: agents not executing a task wait at endpoints that other planned paths must avoid.
\emph{Token Passing with SIPP and Reservation Table (TP-SIPPwRT)}~\citep{ma2019lifelong} replaces TP's low-level planner with Safe Interval Path Planning with Reservation Table, but retains TP's endpoint-resting contract.
Token Passing with Task Swaps (TPTS) adds pre-pickup task reassignment~\citep{ma2019lifelong}; SHARP does not implement that extension and keeps the no-transfer commitment policy for committed tasks.

Our \emph{TP-home-return} and \emph{TP-SIPP-home-return} baselines are diagnostic counterfactuals.
They keep TP's token order and admissible-task filter, assign each agent a fixed return cell at its start, require current$\rightarrow$pickup$\rightarrow$delivery$\rightarrow$return-cell feasibility before commitment, and force agents away from that cell to return before reconsidering tasks.
They are meant to isolate fixed return, not task swaps or a general TP-family dominance claim.
We use \emph{home} only for the analogous fixed return cell in these TP-style counterfactuals; \emph{Haven} denotes SHARP's invariant-carrying owner-only fallback cell.

\begin{table*}[!t]
\centering
\caption{Positioning relative to fallback-based MAPD mechanisms. ``Dense endpoints'' means pickup/delivery locations may lie on narrow passageways; ``parking-limited'' means the guarantee can relax endpoint avoidance only through designated parking/dummy endpoints. Abbreviations: onl.=online, offl.=offline, MG=multi-goal MAPD.}
\label{tab:positioning}
{\footnotesize
\setlength{\tabcolsep}{3pt}
\renewcommand{\arraystretch}{1.08}
\begin{tabularx}{\textwidth}{@{}l c >{\raggedright\arraybackslash}X >{\raggedright\arraybackslash}X >{\raggedright\arraybackslash}X c >{\raggedright\arraybackslash}X@{}}
\hline
\shortstack[c]{Method} &
\shortstack[c]{Setting} &
\shortstack[c]{Fallback\\target} &
\shortstack[c]{Completeness\\assumption} &
\shortstack[c]{Fallback\\cell source} &
\shortstack[c]{Dense endpoints\\in guarantee} &
\shortstack[c]{Comment} \\
\hline
TP~\citep{ma2017lifelong}
& onl.
& endpoint holding
& well-formed MAPD
& non-task endpoints
& no
& Complete online MAPD with endpoint-avoidance \\

\citet{liu2019taskpathplanning}
& offl.
& designated parking endpoint + dummy path
& well-formed / parking-avoidance variant
& designated parking endpoints
& parking-limited
& Offline dummy-path feasibility; no online suffix overwrite \\

\citet{xu2022multigoal}
& MG
& dummy endpoint from shared endpoint pool
& well-formed MG-MAPD
& shared endpoint pool
& no
& Goal-sequence paths with dynamic dummy reassignment \\

SHARP
& onl.
& exclusive Haven (fixed start cell)
& Haven-Reachable + finite releases + progress conditions
& fixed per-agent starts
& yes
& Every busy agent keeps a Haven retreat reservation \\
\hline
\end{tabularx}
}
\end{table*}
For SHARP, the listed progress conditions are framework-level assumptions discharged for the concrete policy in Section~\ref{sec:proof}.

\section{SHARP Coordination Framework}
\label{sec:framework}

\subsection{Setting}

We consider MAPD in discrete time on an undirected graph $G=(V,E)$, where vertices $V$ are traversable locations and edges $E$ are valid bidirectional moves.
For any vertex subset $S \subseteq V$, $G[S]$ denotes the subgraph induced by $S$.
At each timestep $t\in\mathbb{N}_0:=\{0,1,2,\ldots\}$, an agent may either wait at its current vertex or traverse one incident edge.
A collision occurs when two agents occupy the same vertex at the same timestep or traverse the same edge in opposite directions simultaneously.

Let $A$ be the set of agents.
Let $V_{\mathrm{task}}\subseteq V$ be the task endpoint candidate set known to the planner.
Each task $\tau$ has a pickup location $s_{\tau} \in V$, a delivery location $g_{\tau} \in V$, and a release time $r_{\tau} \in \mathbb{N}_0$.
Every released task satisfies $s_{\tau},g_{\tau}\in V_{\mathrm{task}}$.
Before time $r_{\tau}$, task $\tau$ is unreleased.
Once released, it is pending until committed by the planner, then \emph{in-progress}, and finally \emph{completed} when its committed agent reaches $g_{\tau}$.
For both the theorem and the experiments, visiting $s_{\tau}$ or $g_{\tau}$ completes the corresponding pickup or delivery operation immediately.
Deterministic dwell times could be incorporated by adding reserved waits at task vertices, which SIPP can represent through safe intervals, and by enlarging the segment horizons accordingly; that extension is outside the present proof and evaluation.
The fixed-Haven policies analyzed here, including SHARP, do not transfer a committed task to another agent; this is an algorithmic commitment policy rather than a restriction on the MAPD input instance.
Let $\mathrm{pos}_a(t)$ be agent $a$'s position at timestep $t$.

SHARP assigns each agent $a \in A$ an exclusive Haven $\eta(a) \in V$, and we write $H=\eta(A)=\{\eta(a)\mid a\in A\}$.
In the fixed-Haven version studied here, the Haven is the agent's start cell and does not change during execution.
The Haven is exclusive to its owner: agent $a$ may occupy $\eta(a)$, while every other agent treats $\eta(a)$ as blocked.
This owner-only blocking rule is part of the fixed-Haven policy rather than the static instance class defined next.

For the proof and algorithm description, each agent is always in exactly one of three coordination states: \emph{busy}, \emph{retreating}, or \emph{idle}.
A busy agent is executing an assigned task and keeps a reserved path from its current position via pickup and delivery to its Haven.
A retreating agent has no unfinished assigned task and follows the remaining reserved suffix from its current position back to its Haven; this suffix may be overwritten if a new task is accepted.
An idle agent waits at its Haven with no future reservation outside it.
The owner-only Haven rule enforces idle Haven occupancy, so an agent's own Haven wait does not block its planning.

\subsection{Haven-Reachability}

\begin{definition}[Haven-Reachability]
\label{def:haven_reachability}
A MAPD instance is \emph{Haven-Reachable} if there exists a \emph{task-support region} $W\subseteq V\setminus H$ such that all of the following hold:
\begin{enumerate}
\item \textbf{Task-support connectivity}: all task endpoint candidates in $V_{\mathrm{task}}$ lie in $W$, and $G[W]$ is connected.
\item \textbf{Haven access from the support}: for every agent $a \in A$ and every vertex $v \in W$, there exists a simple path from $v$ to $\eta(a)$ whose internal vertices all lie in $W$.
\item \textbf{Assigned Haven uniqueness}: $\eta:A\to H$ is injective, so each assigned Haven has exactly one owner.
\end{enumerate}
We call any such $W$ a witness support region.
\end{definition}

Haven-Reachability is thus a \emph{static} sufficient condition, not a necessary condition for all task-completing schedules.
The support region $W$ may be the whole non-Haven traversable region, but it need not include non-Haven vertices that are irrelevant to tasks and Haven access.
Haven exclusivity is enforced by SHARP's reservation invariant below, not built into $G$.

\paragraph{Relation to well-formedness.}
Well-formed MAPD~\citep{ma2017lifelong} requires paths between endpoints that avoid other endpoints.
Haven-Reachability protects only the per-agent Havens and does not require a fallback pool separate from the start cells selected as Havens, so dense pickup/delivery vertices may lie on corridors as long as they belong to a connected support region that can reach each Haven.
This is related to the parking-avoidance intuition noted by \citet{liu2019taskpathplanning}, but SHARP applies it online with owner-only fixed Havens and a reservation invariant.

\begin{proposition}[Haven-Reachability without well-formedness]
\label{prop:haven_not_wellformed}
Haven-Reachability does not imply well-formedness.
\end{proposition}
\begin{proof}
Take a path $e_1-e_2-e_3$, attach distinct Havens $h_1$ to $e_1$ and $h_2$ to $e_3$, and let $e_1,e_2,e_3$ be task endpoints.
With $W=\{e_1,e_2,e_3\}$, the instance is Haven-Reachable: $G[W]$ is connected, every endpoint can reach each Haven through $W$ until the final Haven edge, and the Havens have distinct owners.
It is not well-formed, because any path from endpoint $e_1$ to endpoint $e_3$ passes through endpoint $e_2$.
\end{proof}

\subsection{Reservation invariant}

SHARP maintains a space-time reservation table $\mathcal{R}$ for future routes of busy and retreating agents, and plans with SIPP~\citep{phillips2011sipp}.
Idle Haven occupancy is handled by the owner-only blocking rule rather than by inserting infinite waits into $\mathcal{R}$.
SIPP searches time-stamped paths by reasoning over intervals during which each vertex is safe to occupy, while avoiding existing reservations.
The crucial invariant is:

\begin{definition}[Safe-Haven Reservation Invariant]
\label{def:reservation_invariant}
At every timestep:
\begin{enumerate}
\item each busy agent reserves a complete collision-free remaining path from its current position through any unvisited task waypoint(s), namely pickup if not yet reached and then delivery, and finally to $\eta(a)$,
\item each retreating agent reserves a collision-free suffix from its current position to its own Haven,
\item reservations contain no vertex or edge-swap conflicts,
\item no agent ever occupies another agent's Haven, and
\item idle agents remain at their own Havens.
\end{enumerate}
\end{definition}

Fix a witness support region $W$ for Haven-Reachability.
We adopt a \emph{per-call relative horizon} $T_{\max}$ for each SIPP invocation inside \textsc{PlanFullPath}, measured from the start time of that segment, and set it to satisfy
\[
T_{\max} \ge \mathrm{diam}(G[W]) + 1,
\]
where $\mathrm{diam}(\cdot)$ denotes graph diameter.
In quiescent configurations, this suffices because the pickup$\rightarrow$delivery segment lies in $G[W]$ and each Haven-touching segment needs at most one additional incident edge from $W$ into the Haven.
This is a per-segment SIPP bound; outside quiescence, SHARP retries failed validations while preserving the current retreat reservation.

\subsection{Task assignment and path validation}

SHARP uses a nearest-pickup assignment rule among agents that are idle or currently retreating.
In each round, every candidate agent proposes its nearest pending task under the static shortest-path distance $\text{dist}(\cdot,\cdot)$ on the traversable graph.
This distance is only a ranking heuristic: it ignores time-space reservations and Haven occupancy restrictions, whereas feasibility is decided only by SIPP validation with all foreign Havens statically blocked for that agent.
The algorithm then commits the feasible pair with minimum pickup distance.
Ties can be broken by any fixed total order over agent and task identifiers; the proof does not depend on the particular order.
In the pseudocode, \textsc{SIPP}$(u,v,t,\mathcal{R},B,T_{\max})$ denotes a SIPP search from $u$ to $v$ starting at time $t$ with reservations $\mathcal{R}$, blocked vertices $B$, and per-call horizon $T_{\max}$.
\textsc{ArrivalTime}$(\pi)$ returns the final timestamp of path $\pi$, and \textsc{Concat}$(\pi_1,\pi_2,\pi_3)$ concatenates time-consistent path segments while removing duplicated boundary states.
A candidate assignment is accepted only if SIPP can validate the full path to pickup, delivery, and Haven.
For a retreating candidate, this validation temporarily removes only that agent's own future retreat suffix from the reservation table; upon commitment, that suffix is overwritten by the new full reservation.
If validation fails, the agent follows its already reserved retreat path and the assignment can be retried later.
In Algorithm~\ref{alg:sharp}, \textsc{CandidateReservations}$(\mathcal{R},a)$ denotes $\mathcal{R}$ with only agent $a$'s own future retreat suffix temporarily ignored when $a$ is retreating.
This prevents the candidate's existing retreat reservation from being treated as a collision with its own replacement route; all other agents' reservations and all foreign Havens remain blocking constraints.
Outside quiescence, this nearest-pickup proposal rule is opportunistic.
A candidate may reject its nearest task even when a farther pending task would validate under the current reservations.
This can affect efficiency, but not the finite-release completeness proof, which relies on progress only after quiescence.

\begin{algorithm}[t]
\caption{SHARP (fixed-Haven version)}
\label{alg:sharp}
\begin{algorithmic}[1]
\STATE Initialize $t \leftarrow 0$, $\mathcal{R} \leftarrow \emptyset$
\WHILE{the run is active}
    \STATE $P \leftarrow$ released pending tasks
    \STATE $I \leftarrow$ idle agents $\cup$ retreating agents
    \WHILE{$I \neq \emptyset$ \AND $P \neq \emptyset$}
        \STATE $c^\star \leftarrow \bot$; $d^\star \leftarrow \infty$
        \FOR{$a \in I$}
            \STATE $\tau \leftarrow \arg\min_{\tau' \in P}\text{dist}(\mathrm{pos}_a(t), s_{\tau'})$
            \STATE $\mathcal{R}_a \leftarrow \textsc{CandidateReservations}(\mathcal{R},a)$
            \STATE $(f,\pi)\leftarrow \textsc{PlanFullPath}(a,\tau,\mathcal{R}_a,T_{\max})$
            \IF{$f$ and $\text{dist}(\mathrm{pos}_a(t), s_{\tau}) < d^\star$}
                \STATE $c^\star \leftarrow (a,\tau,\pi)$; $d^\star \leftarrow \text{dist}(\mathrm{pos}_a(t), s_{\tau})$
            \ENDIF
        \ENDFOR
        \IF{$c^\star = \bot$}
            \STATE \textbf{break}
        \ENDIF
        \STATE Unpack $(a,\tau,\pi) \leftarrow c^\star$
        \STATE If $a$ is retreating, remove only $a$'s future retreat suffix from $\mathcal{R}$
        \STATE Commit $\pi$ to $\mathcal{R}$, assign $\tau$ to $a$, and mark $a$ busy
        \STATE Remove $a$ and $\tau$ from $I$ and $P$
    \ENDWHILE
    \STATE Execute one timestep along reserved paths; $t \leftarrow t+1$
\ENDWHILE
\end{algorithmic}
\end{algorithm}

\begin{algorithm}[t]
\caption{\textsc{PlanFullPath}$(a,\tau,\mathcal{R},T_{\max})$}
\label{alg:plan_full_path}
\begin{algorithmic}[1]
\STATE $t_{\mathrm{start}} \leftarrow$ current timestep
\STATE $x_0 \leftarrow \mathrm{pos}_a(t_{\mathrm{start}})$
\STATE $B_a \leftarrow H \setminus \{\eta(a)\}$
\STATE $(ok_1,\pi_1)\leftarrow \textsc{SIPP}(x_0,s_{\tau},t_{\mathrm{start}},\mathcal{R},B_a,T_{\max})$
\IF{$\neg ok_1$}
    \STATE \textbf{return} $(\text{false}, \emptyset)$
\ENDIF
\STATE $t_1 \leftarrow \textsc{ArrivalTime}(\pi_1)$
\STATE $(ok_2,\pi_2)\leftarrow \textsc{SIPP}(s_{\tau},g_{\tau},t_1,\mathcal{R},B_a,T_{\max})$
\IF{$\neg ok_2$}
    \STATE \textbf{return} $(\text{false}, \emptyset)$
\ENDIF
\STATE $t_2 \leftarrow \textsc{ArrivalTime}(\pi_2)$
\STATE $(ok_3,\pi_3)\leftarrow \textsc{SIPP}(g_{\tau},\eta(a),t_2,\mathcal{R},B_a,T_{\max})$
\IF{$\neg ok_3$}
    \STATE \textbf{return} $(\text{false}, \emptyset)$
\ENDIF
\STATE \textbf{return} $(\text{true}, \text{Concat}(\pi_1,\pi_2,\pi_3))$
\end{algorithmic}
\end{algorithm}

\paragraph{Completeness scope of \textsc{PlanFullPath}.}
\textsc{PlanFullPath} is a conservative sequential validator and does not backtrack from an accepted earlier segment.
Thus a failed validation outside quiescence does not imply that no feasible route exists; it only means SHARP declines that candidate and preserves the agent's current retreat reservation.
The proof only needs successful validation in quiescent configurations, where future reservations outside Havens are empty and each leg can be found independently.

\begin{proposition}[Per-timestep complexity]
\label{prop:complexity}
Let $|P|$ be the number of pending tasks, $C_{\text{SIPP}}(T_{\max})$ the cost of one SIPP call up to horizon $T_{\max}$, and $C_{\text{full}}(T_{\max})$ the cost of one full-path validation, consisting of at most three SIPP calls, each with per-call horizon $T_{\max}$.
Then SHARP requires
\[
O\!\left(\min(|A|,|P|)\,|A|\,\big(|P|+C_{\text{full}}(T_{\max})\big)\right)
\]
time per timestep, plus distance-precomputation cost.
This follows because at most $\min(|A|,|P|)$ pairs can be committed per timestep, each round scans up to $|A|$ candidate agents, and each scan performs a nearest-task search over $|P|$ pending tasks plus one full-path SIPP validation.
\end{proposition}
Thus SHARP trades more centralized planning work for the stronger reservation invariant; Section~\ref{sec:experiments} reports the measured cost.

\section{Completeness Under Fixed Havens}
\label{sec:proof}

\begin{definition}[Finite-release completeness]
\label{def:finite_release_complete}
A policy is \emph{finite-release complete} under a stated set of assumptions if, whenever only finitely many tasks are released, every released task is delivered in finite time.
\end{definition}

\begin{definition}[Quiescent Configuration]
\label{def:quiescent}
The system is in a \emph{quiescent configuration} at time $t$ if every agent is at its own Haven and no future reservation leaves that Haven.
\end{definition}

\begin{definition}[Feasible agent-task pair]
\label{def:feasible_pair}
At time $t$, a pending task $\tau$ is \emph{feasible} for agent $a$ if there exists a finite collision-free reserved route from $\mathrm{pos}_a(t)$ through $s_{\tau}$ and $g_{\tau}$ to $\eta(a)$, respecting current reservations and treating all foreign Havens as blocked for $a$.
In SHARP, \textsc{PlanFullPath} is used as a sufficient validation procedure for such pairs; it is not claimed to be a complete feasibility checker outside quiescent configurations.
The SHARP corollary below states the extra SIPP horizon condition under which this validator finds the quiescent feasible pairs needed below.
\end{definition}

\begin{definition}[Assignment-Progress Property]
\label{def:assignment_progress}
A task-assignment strategy satisfies the \emph{Assignment-Progress Property} if, whenever at least one feasible agent-task pair exists in a quiescent configuration, the strategy assigns at least one such pair.
\end{definition}

\begin{definition}[Route-Execution Progress Condition]
\label{def:route_progress}
A fixed-Haven execution policy satisfies the \emph{Route-Execution Progress Condition} if, once no new task is committed after time $t$, every non-idle agent executes the next reserved state of a finite valid suffix at each timestep until it reaches its Haven; moreover, if the suffix contains committed task waypoints, those pickup/delivery waypoints are reached before the final Haven.
This rules out failed validations, replanning attempts, or route-overwrite opportunities that defer committed routes through delivery and Haven forever.
\end{definition}

Corollary~\ref{cor:assignment_progress} below shows that SHARP's nearest-pickup rule satisfies the Assignment-Progress Property.

\begin{lemma}[Invariant preservation]
\label{lem:invariant}
SHARP preserves the Safe-Haven Reservation Invariant at every timestep.
\end{lemma}
\begin{proof}
Initially, all agents are idle at their Havens.
A path is committed only if SIPP validates a collision-free concatenation ending at the owner's Haven while treating foreign Havens as blocked, so committed paths preserve Haven exclusivity and collision freedom.
During execution, agents move only along validated reservations, and once delivery finishes the remaining suffix to Haven becomes the retreat reservation.
When a retreating agent accepts a new task, its old suffix is removed only through immediate replacement by another validated full path to the same Haven.
Therefore the invariant is preserved inductively.
\end{proof}

\begin{lemma}[Feasibility in quiescence]
\label{lem:quiescence_feasible}
Assume Haven-Reachability with witness support region $W$.
In quiescence, each pending task is feasible for each agent.
\end{lemma}
\begin{proof}
In quiescence, all vertices in $W$ are free in the future because all agents wait at their own Havens.
Haven access from the support and undirectedness give a path between the agent's Haven and each task vertex whose non-Haven vertices lie in $W$; task-support connectivity connects pickup to delivery inside $G[W]$.
Concatenating these paths gives a finite route from the agent's current Haven through pickup and delivery back to its Haven.
The route avoids foreign Havens and has no conflicting future reservation.
\end{proof}

\begin{corollary}[SHARP satisfies Assignment-Progress]
\label{cor:assignment_progress}
Assume Haven-Reachability with witness support region $W$, $T_{\max} \ge \mathrm{diam}(G[W]) + 1$, and SIPP completeness up to this per-call relative horizon with respect to the reservation table and static foreign-Haven blocking constraints.
Then fixed-Haven SHARP's nearest-pickup assignment rule satisfies the Assignment-Progress Property.
\end{corollary}
\begin{proof}
Consider any quiescent configuration with at least one pending task.
Then every agent is idle at its Haven, so in Algorithm~\ref{alg:sharp} we have $I=A$.
By Lemma~\ref{lem:quiescence_feasible}, every pending task is feasible for every agent.
For the task selected by each $a \in I$, the pickup$\rightarrow$delivery segment has length at most $\mathrm{diam}(G[W])$, and each Haven-touching segment has length at most $\mathrm{diam}(G[W]) + 1$.
Since each SIPP call uses horizon $T_{\max}$ from its own segment start time and blocks only foreign Havens, \textsc{PlanFullPath} returns a feasible full path for that selected task.
Thus some feasible candidate is generated and $c^\star \neq \bot$.
The commitment step therefore commits at least one feasible agent-task pair.
Thus the Assignment-Progress Property holds.
\end{proof}

\begin{lemma}[Convergence to quiescence under SHARP]
\label{lem:convergence}
If no further tasks are assigned after some timestep $t$, then SHARP reaches a quiescent configuration in finite time.
\end{lemma}
\begin{proof}
By Lemma~\ref{lem:invariant}, every busy or retreating agent already owns a finite collision-free reserved path suffix ending at its Haven.
In Algorithm~\ref{alg:sharp}, busy agents are never reconsidered for task assignment, and a retreat suffix is overwritten only when a retreating agent accepts a newly validated full path to a new task and back to the same Haven.
Therefore, if no further tasks are assigned after time $t$, no reserved suffix is overwritten after time $t$.
The execution step of Algorithm~\ref{alg:sharp} executes the next state of each non-idle agent's finite reserved suffix at each timestep, so each reaches its Haven in finite time, after which no future reservation leaves it.
\end{proof}

The role of the framework theorem is to separate the generic progress argument from the graph-structural content.
The nontrivial step is supplied by Lemma~\ref{lem:quiescence_feasible}: after quiescence, Haven-Reachability makes every pending task feasible for every Haven owner because the support region is clear and connected.

\begin{theorem}[Finite-release completeness of the fixed-Haven framework]
\label{thm:framework_completeness}
Assume Haven-Reachability with witness support region $W$, finitely many task releases, and the no-transfer commitment policy defined in Section~\ref{sec:framework}.
Any fixed-Haven policy that maintains the Safe-Haven Reservation Invariant, satisfies the Assignment-Progress Property, and satisfies the Route-Execution Progress Condition is finite-release complete.
\end{theorem}
\begin{proof}
Suppose some released task remains unfinished forever.
A committed task receives, by the invariant, a finite reserved route whose delivery waypoint precedes the final Haven.
Because task releases are finite, each commitment consumes one previously pending task, and the no-transfer policy prevents transfer or recommitment of committed tasks, the total number of commitments is finite.
Let $t^\star$ be after the last release and the last commitment.
By Route-Execution Progress, every task committed by $t^\star$ reaches its delivery waypoint in finite time, and all non-idle agents then reach their Havens in finite time.
Thus any task that remains unfinished forever must be pending in a quiescent configuration.
Lemma~\ref{lem:quiescence_feasible} then makes every pending task feasible for every agent, so Assignment-Progress forces another commitment, contradicting the choice of $t^\star$.
\end{proof}

\begin{corollary}[Finite-release completeness of fixed-Haven SHARP]
\label{cor:sharp_completeness}
Assume Haven-Reachability with witness support region $W$, finitely many task releases, the no-transfer commitment policy defined in Section~\ref{sec:framework}, a per-call relative horizon $T_{\max} \ge \mathrm{diam}(G[W]) + 1$, and SIPP completeness up to that horizon with respect to reservations and static foreign-Haven blocking.
Then fixed-Haven SHARP completes every released task in finite time.
\end{corollary}
\begin{proof}
Lemma~\ref{lem:invariant} gives the Safe-Haven Reservation Invariant, Corollary~\ref{cor:assignment_progress} gives Assignment-Progress, and Lemma~\ref{lem:convergence} gives the Route-Execution Progress Condition for SHARP.
The claim follows from Theorem~\ref{thm:framework_completeness}.
\end{proof}

\section{Warehouse-Inspired Experiments}
\label{sec:experiments}

We report a main robustness sweep plus two targeted diagnostics: a TP-style fixed-home-return counterfactual and a no-overwrite variant of SHARP.

\subsection{Setup}

\begin{figure}[!t]
\centering
\begin{subfigure}[b]{0.48\columnwidth}
\centering
\includegraphics[width=\textwidth]{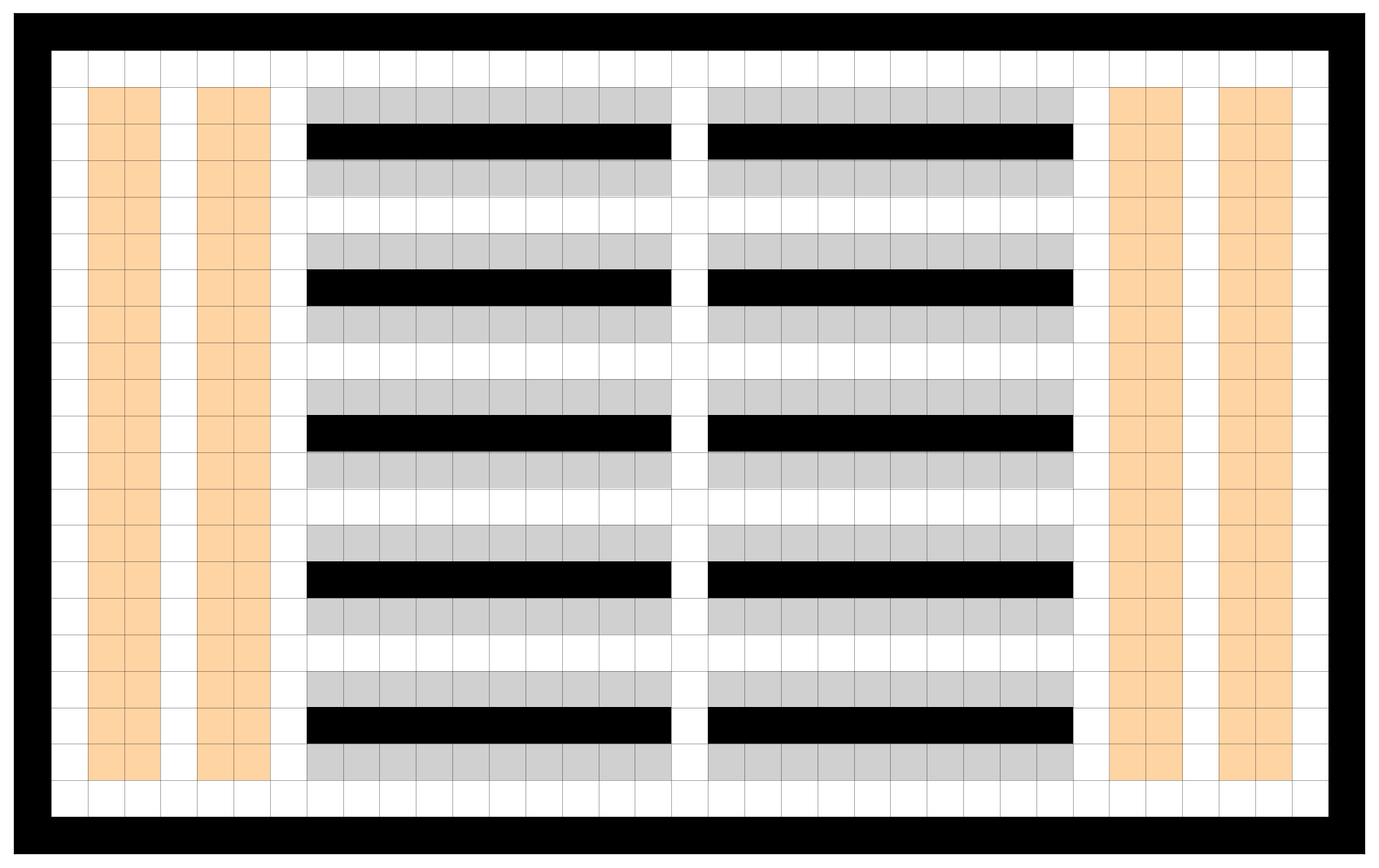}
\caption{well-formed}
\label{fig:map_wf}
\end{subfigure}
\hfill
\begin{subfigure}[b]{0.48\columnwidth}
\centering
\includegraphics[width=\textwidth]{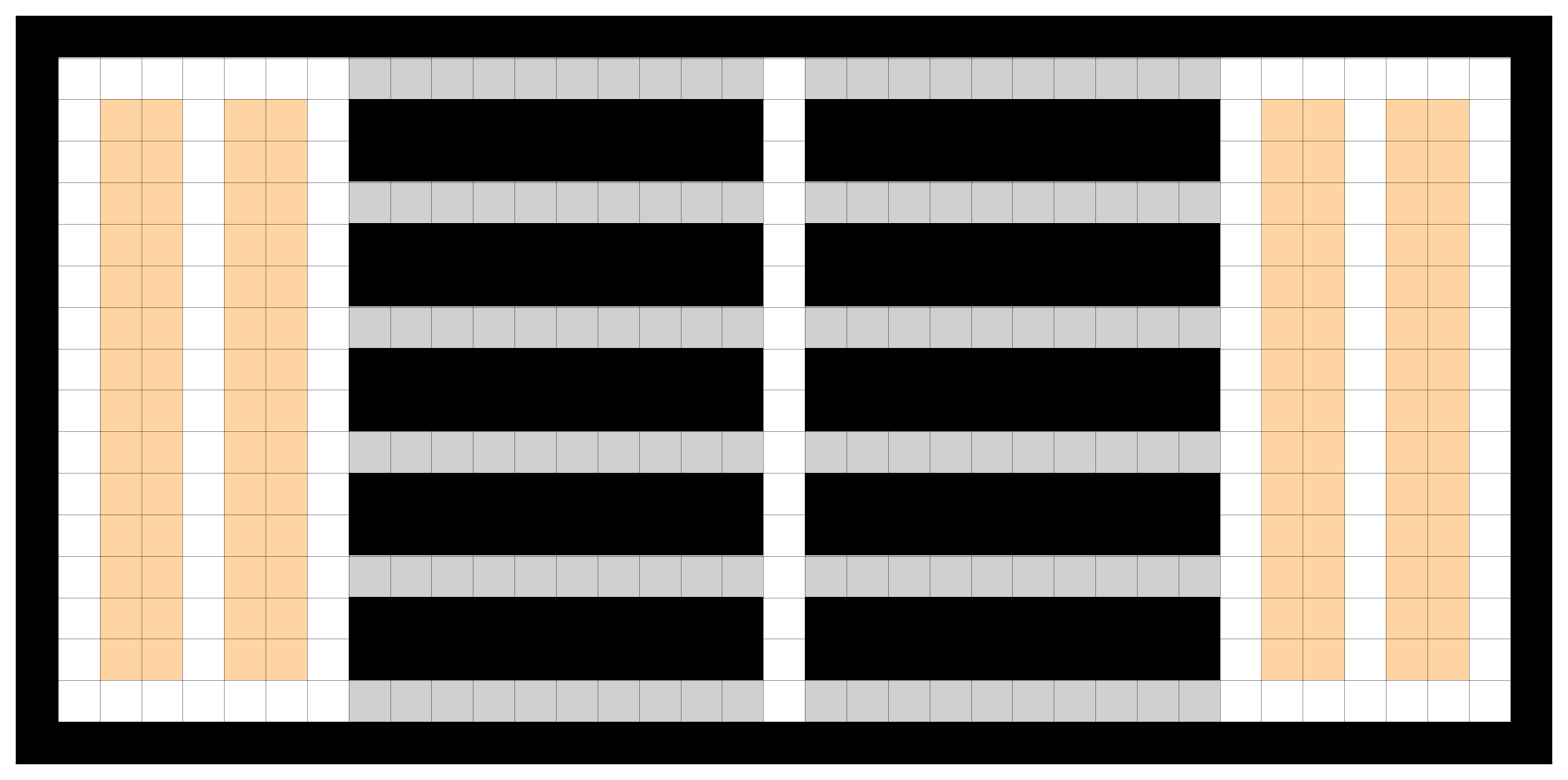}
\caption{narrow-bi}
\label{fig:map_nb}
\end{subfigure}

\vspace{0.5em}

\begin{subfigure}[b]{0.48\columnwidth}
\centering
\includegraphics[width=\textwidth]{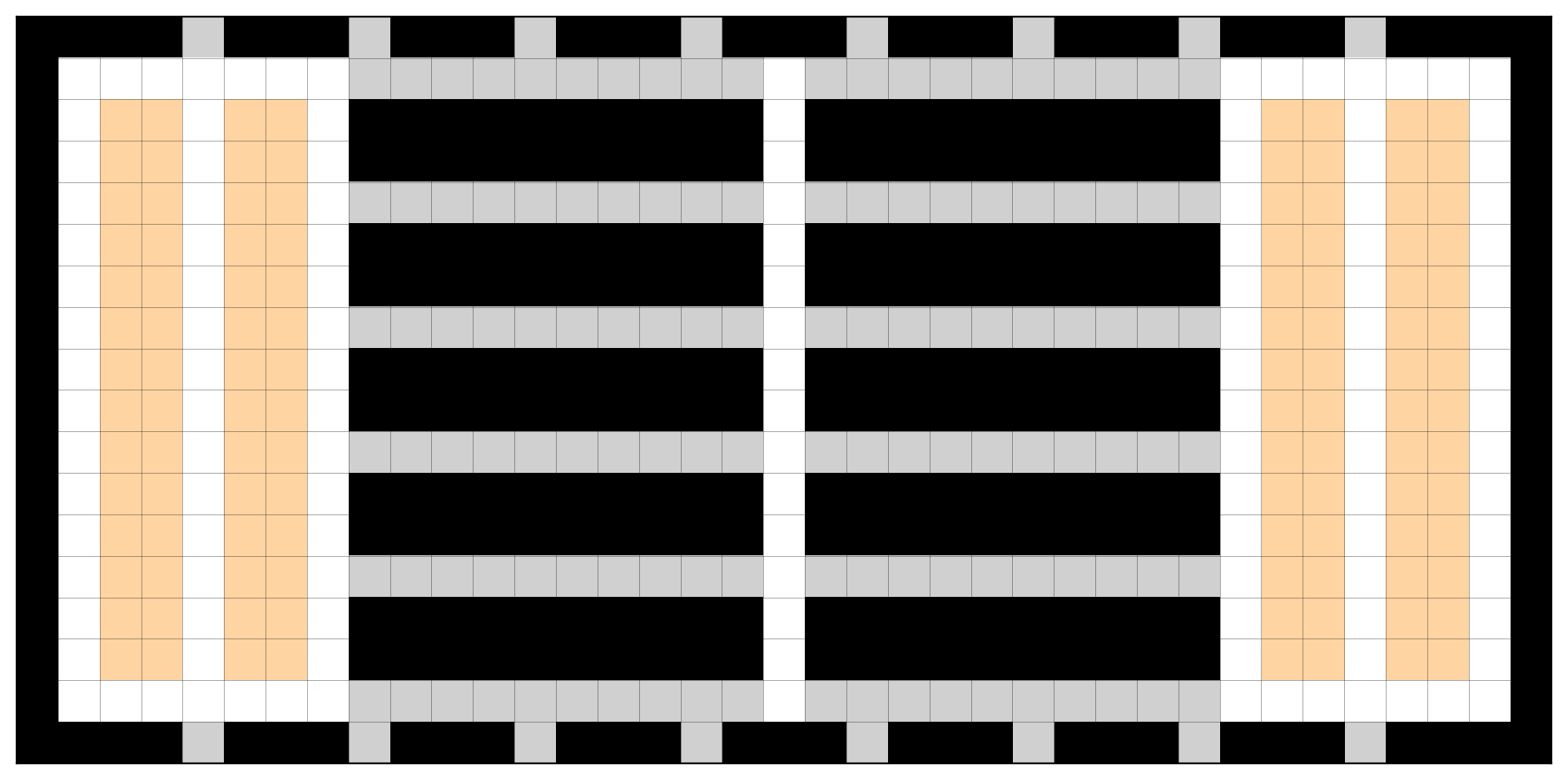}
\caption{narrow-bi-dead}
\label{fig:map_nbd}
\end{subfigure}
\hfill
\begin{subfigure}[b]{0.48\columnwidth}
\centering
\includegraphics[width=\textwidth]{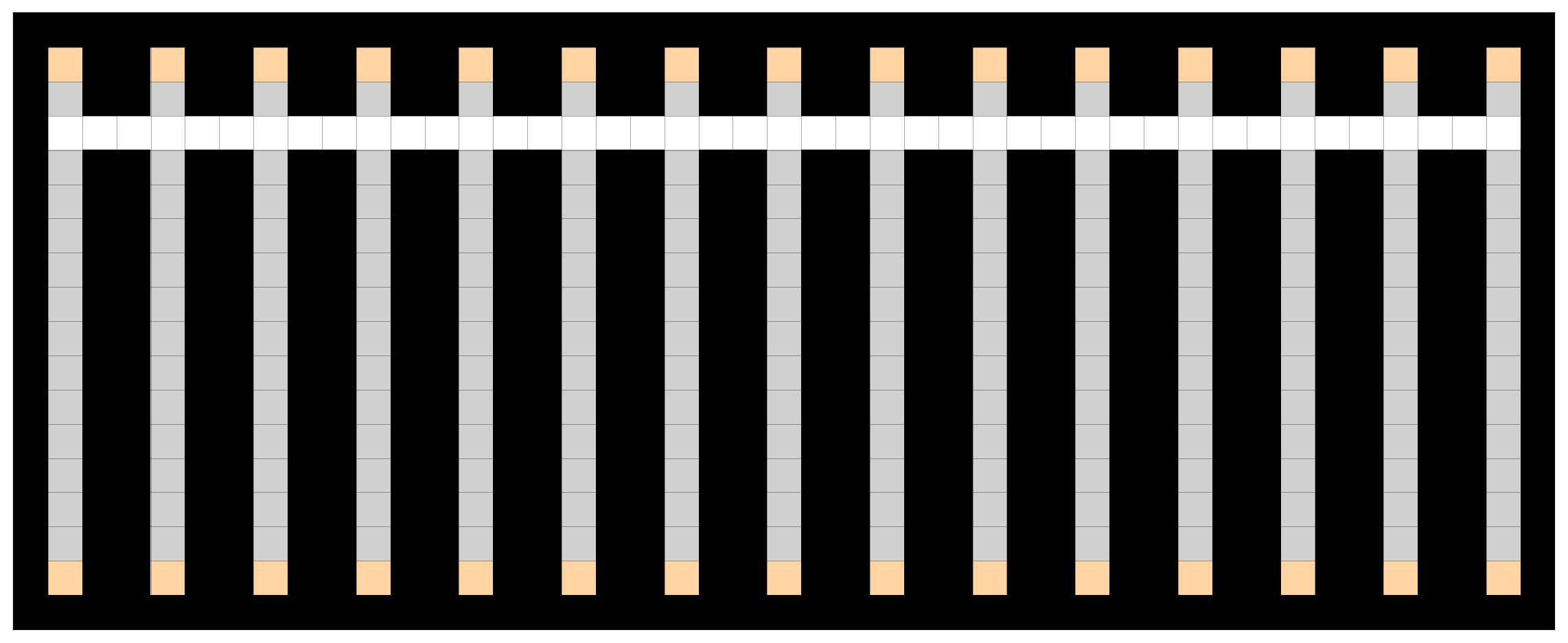}
\caption{tree}
\label{fig:map_tree}
\end{subfigure}
\caption{Map classes. Narrow-bi abbreviates narrow-biconnected, and narrow-bi-dead abbreviates narrow-biconnected with dead ends. Orange cells denote start/Haven candidate locations; in a given run only the sampled start cells belong to $H$. The same sampled start cells are fixed homes for TP-home-return variants. Gray cells are task endpoints, white cells are traversable non-task cells, and black cells are obstacles.}
\label{fig:maps}
\end{figure}

We use four map classes (Figure~\ref{fig:maps}):
\begin{itemize}
\item \textbf{well-formed} ($38 \times 23$): the warehouse-style benchmark of \citet{ma2017lifelong}, satisfying the endpoint-avoidance condition used by TP;
\item \textbf{narrow-bi} ($38 \times 18$): a biconnected but non-well-formed corridor-shrunk variant;
\item \textbf{narrow-bi-dead} ($38 \times 18$): the same with depth-1 workstation/staging dead ends, matching the PIBTTP-TA main-area-plus-attached-tree topology; and
\item \textbf{tree} ($45 \times 18$): a tree-like storage abstraction motivated by warehouse storage/retrieval and MAPD studies~\citep{azadeh2017robotized, roy2017multitier, Iida2023, Hirayama2025}.
\end{itemize}

For each configuration-seed pair, we sample starts/Havens uniformly from orange cells, tasks from the gray task endpoint candidates with pickup $\neq$ delivery, and reuse the identical instance across compared methods.
Unused orange cells remain traversable.
In the experiments we use the conservative witness $W=V\setminus H$; for every sampled Haven set we verified $V_{\mathrm{task}}\subseteq W$, $G[W]$ connectivity, Haven access from $W$, and Haven uniqueness.

The main sweep uses agent counts $\{5,10,15,20,25,30\}$, rates $\lambda \in \{0.5,1.0,\ldots,3.0\}$, 100 runs per configuration, and algorithms TP, PIBT, PIBTTP-TA, and SHARP.
The home-return sweep uses 30 agents, $\lambda \in \{1.0,3.0\}$, well-formed/tree maps, 30 paired seeds, and TP, TP-home-return, TP-SIPP-home-return, and SHARP.
For TP-home-return and TP-SIPP-home-return, each agent's sampled start cell is its fixed home.
The no-overwrite variant uses the hardest tree condition (30 agents, $\lambda=3.0$), 30 paired seeds, and SHARP with or without mid-retreat overwrite.
It keeps full-route validation and a fleet-wide scan over its eligible set, but restricts that set to idle agents: a retreating agent must reach its Haven before accepting another task.
The non-well-formed cases stress layouts beyond endpoint-avoidance or biconnectivity assumptions.
All sweeps use finite releases at timesteps 0 through 300 inclusive with expected rate $\lambda$ (expected total $301\lambda$): each timestep releases $\lfloor \lambda \rfloor$ guaranteed tasks plus one extra task with probability $\lambda - \lfloor \lambda \rfloor$.

A run is counted as successful only if all released tasks are delivered before any stop condition is triggered.
Runs stop at 10,000 timesteps, 1,000 timesteps without assignment/pickup/delivery, or an external 1800~s wall-clock cap; no reported run hit the wall-clock cap.
SHARP uses the theorem's support-region-relative SIPP horizon; with $W=V\setminus H$, the maximum sampled value of $\mathrm{diam}(G[W])+1$ is 67 steps.
TP and TP-home-return use TP-style Cooperative A*; TP-SIPP-home-return keeps TP-style token and fixed-home-return policies but uses SIPP; PIBT-family methods do not use SIPP.
PIBT-family baselines receive the same released-task stream through a no-transfer MAPD wrapper: an available agent targets the nearest reachable pending pickup under static distance, commits only when it reaches pickup, and then targets delivery.
A provisional pickup target remains pending; if multiple agents target it, the first arrival commits the task and the others reselect.
PIBT or PIBTTP-TA supplies the one-step collision-avoidance and priority-inheritance motion policy for those current goals.
The 1,000-step stall rule is an experimental watchdog only and is not part of the completeness theorem.
TP-SIPPwRT would mainly test a low-level planner substitution under TP's endpoint-resting contract; our TP-SIPP-home-return variant instead keeps fixed-home return and serves as a diagnostic, not an exhaustive TP-family replacement.

We report success, makespan (completion time of the last delivered task), service time (release-to-delivery latency), per-step runtime, and run-level runtime for counterfactuals; for the no-overwrite variant we also report total travel distance, measured as summed edge traversals over all agents in a successful run.
Task metrics are computed over successful runs only.
For each task metric $q$ (makespan or service time), map $m$, plotted condition $x$, and method $b$, let $q_{m,x,b}$ be the successful-run mean.
Main-sweep task plots use
\[
\widetilde{q}_{m,x,b}=\frac{q_{m,x,b}}{q^{\mathrm{ref}}_{m,q}},
\]
where $q^{\mathrm{ref}}_{m,q}$ is the corresponding metric value of the best successful non-SHARP baseline at the map-specific anchor for that plot family ($\lambda=3.0$ for rate sweeps, 5 agents for agent-count sweeps), or SHARP if no non-SHARP baseline succeeds there.
These task metrics must be read after conditioning on success rate.
Non-well-formed absolute task metrics are deployment-sensitive, so we report ratios, success, per-step runtime, and planner effort there; absolute task values appear only for the Ma et al. well-formed benchmark.
Paired counterfactuals use Friedman tests followed by Bonferroni-corrected Wilcoxon tests, and the no-overwrite variant uses one-sided paired Wilcoxon tests.
All methods were implemented in Python and evaluated on identical task sequences generated from fixed random seeds on an AMD Threadripper PRO workstation with 256 GB RAM.

\begin{figure*}[!t]
\centering
\includegraphics[width=0.95\textwidth]{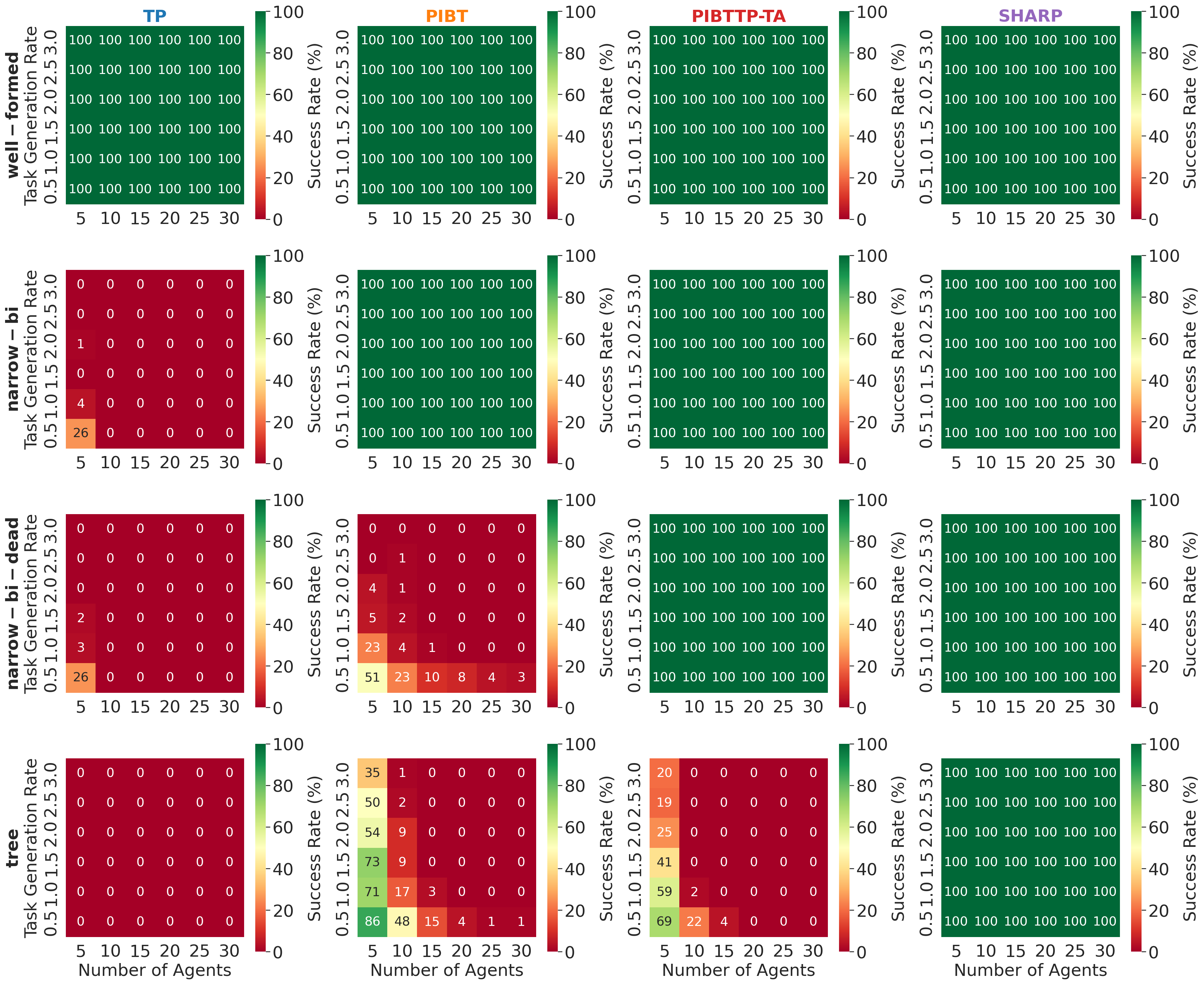}
\caption{Success-rate heatmap for the main sweep. Axes are agent count (x) and task-generation rate $\lambda$ (y). The map panels move from the common well-formed regime to increasingly restrictive narrow, dead-end, and tree layouts. Among the representative TP and PIBT-family MAPD baselines in the main sweep, only SHARP attains 100\% success on all tested configurations.}
\label{fig:success_rate}
\end{figure*}

\paragraph{Comparison scope.}
The main sweep is a stress-test comparison against representative TP and PIBT-family MAPD baselines under our sampled MAPD task distribution, not a theorem-to-theorem dominance claim.
SHARP differs from TP-family baselines by fleet-wide candidate evaluation and retreat-suffix overwrite, whereas TP-family methods keep sequential token ordering~\citep{ma2017lifelong, ma2019lifelong}.
PIBT and PIBTTP-TA are instead distributed/local-communication methods~\citep{okumura2019priority, fujitani2022priority}.
PIBT-family failures outside their intended graph and task-placement regimes should be read as applicability-boundary evidence.
Our PIBTTP-TA implementation follows the temporary-priority and temporary-avoidance motion rules, and the narrow-bi-dead map matches the intended main-area-plus-attached-tree topology, but the sampled task stream is not filtered to exactly satisfy all theorem-side task-placement restrictions of \citet{fujitani2022priority}.
We omit direct runtime comparisons to \citet{liu2019taskpathplanning} and \citet{xu2022multigoal}, because they study offline or multi-goal variants.

\subsection{Main Sweep: Coordination Robustness and Task Performance}

Figure~\ref{fig:success_rate} is the main coordination-robustness result.
On the well-formed map, all four methods remain at 100\% success on all 36 configurations, consistent with TP's well-formed guarantee and the empirical robustness of the PIBT-family methods on this benchmark.
Off that regime the empirical separation becomes clear: TP degrades on narrow-bi maps, PIBT also loses robustness on narrow-bi-dead maps, and on the tree map SHARP is the only method among the representative main-sweep baselines with 100\% success across all tested loads and agent counts.
No run hit the external 1800~s wall-clock cap, only nine unsuccessful runs reached the 10,000-timestep horizon, and 9,911 main-sweep runs triggered the 1,000-step stall watchdog; the heatmap is therefore not a wall-clock-budget artifact.

Figures~\ref{fig:service_time_pair} and~\ref{fig:makespan_pair} complement the success-rate heatmap with normalized task-performance trends.
On well-formed and narrow-bi-dead layouts, SHARP remains competitive with the strongest successful baseline despite reserving a return to Haven.
On tree layouts, these task metrics are secondary because they are conditioned on successful runs only; the main message there is structural robustness at much higher centralized planning cost.
Consequently, tree-layout task ratios should not be read as quality comparisons against failed baselines; they summarize successful-run behavior after the success-rate result has been considered.

\begin{figure*}[!tbp]
\centering
\begin{subfigure}[t]{0.48\textwidth}
\centering
\includegraphics[width=\linewidth]{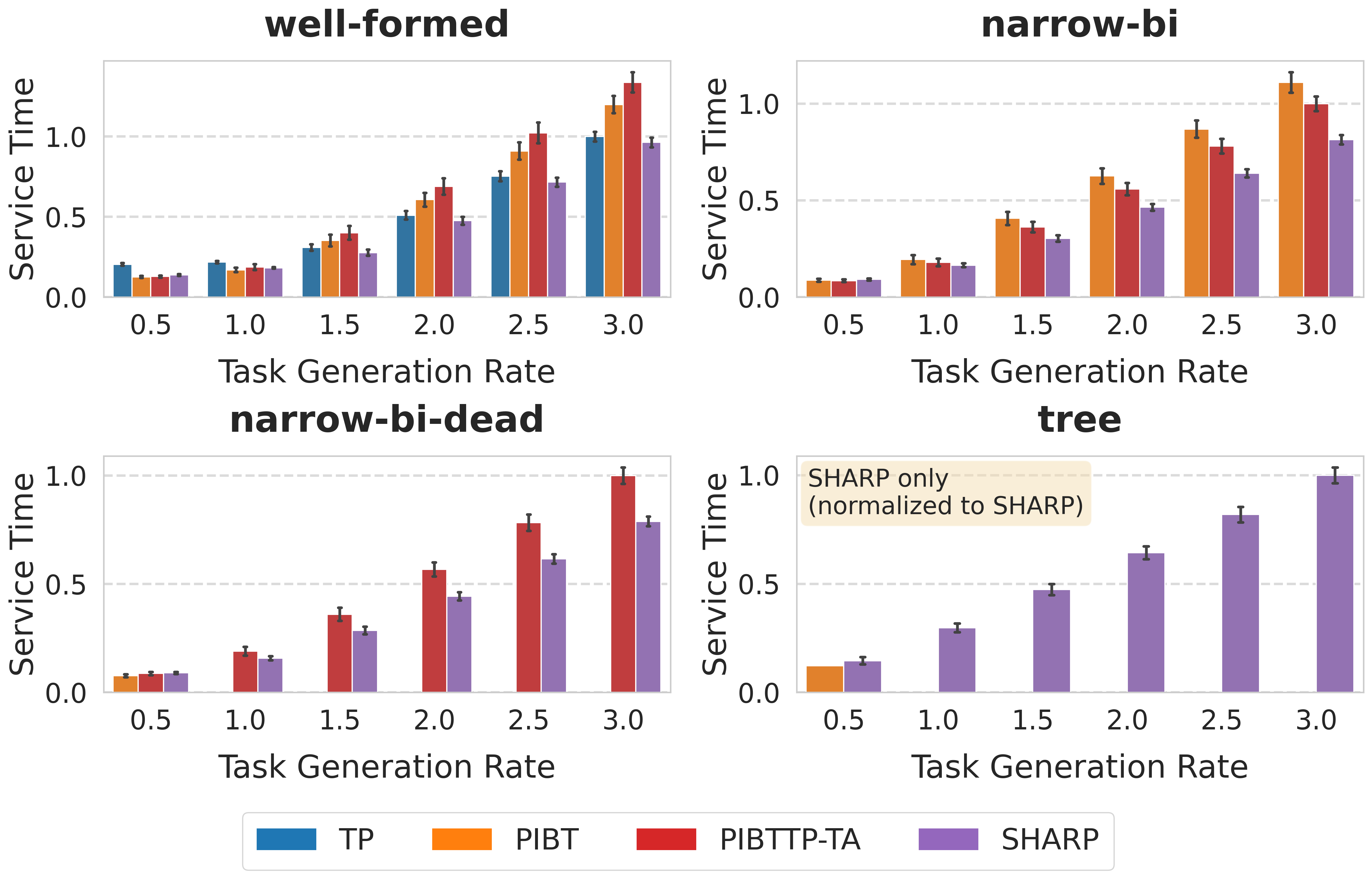}
\caption{Vs. task generation rate (30 agents).}
\label{fig:service_time}
\end{subfigure}
\hfill
\begin{subfigure}[t]{0.48\textwidth}
\centering
\includegraphics[width=\linewidth]{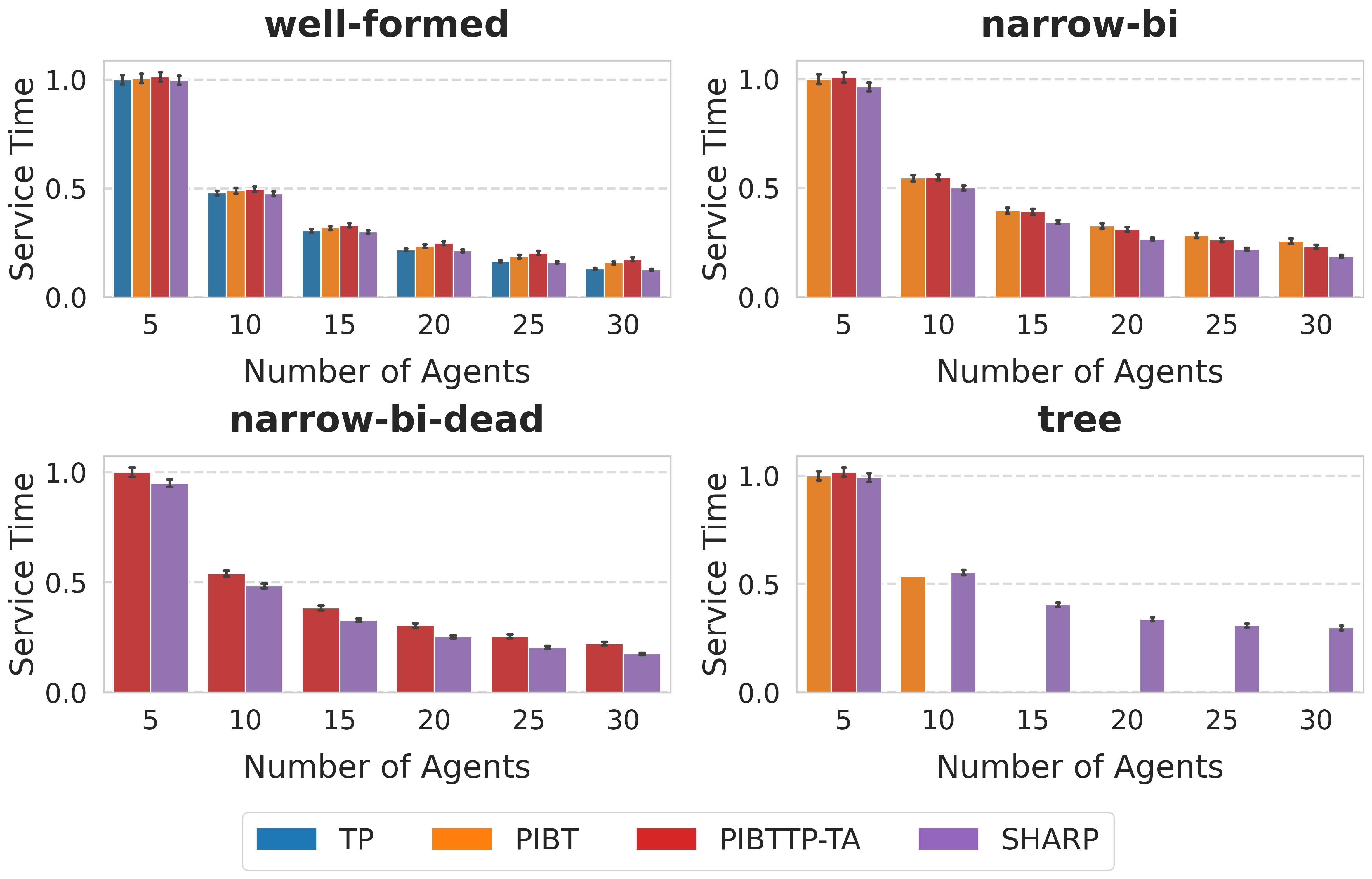}
\caption{Vs. number of agents ($\lambda=3.0$).}
\label{fig:service_time_agents}
\end{subfigure}
\caption{Normalized service time. Bars show successful-run means, error bars show one standard deviation, and missing bars mean no successful runs.}
\label{fig:service_time_pair}
\end{figure*}

\begin{figure*}[!tbp]
\centering
\begin{subfigure}[t]{0.48\textwidth}
\centering
\includegraphics[width=\linewidth]{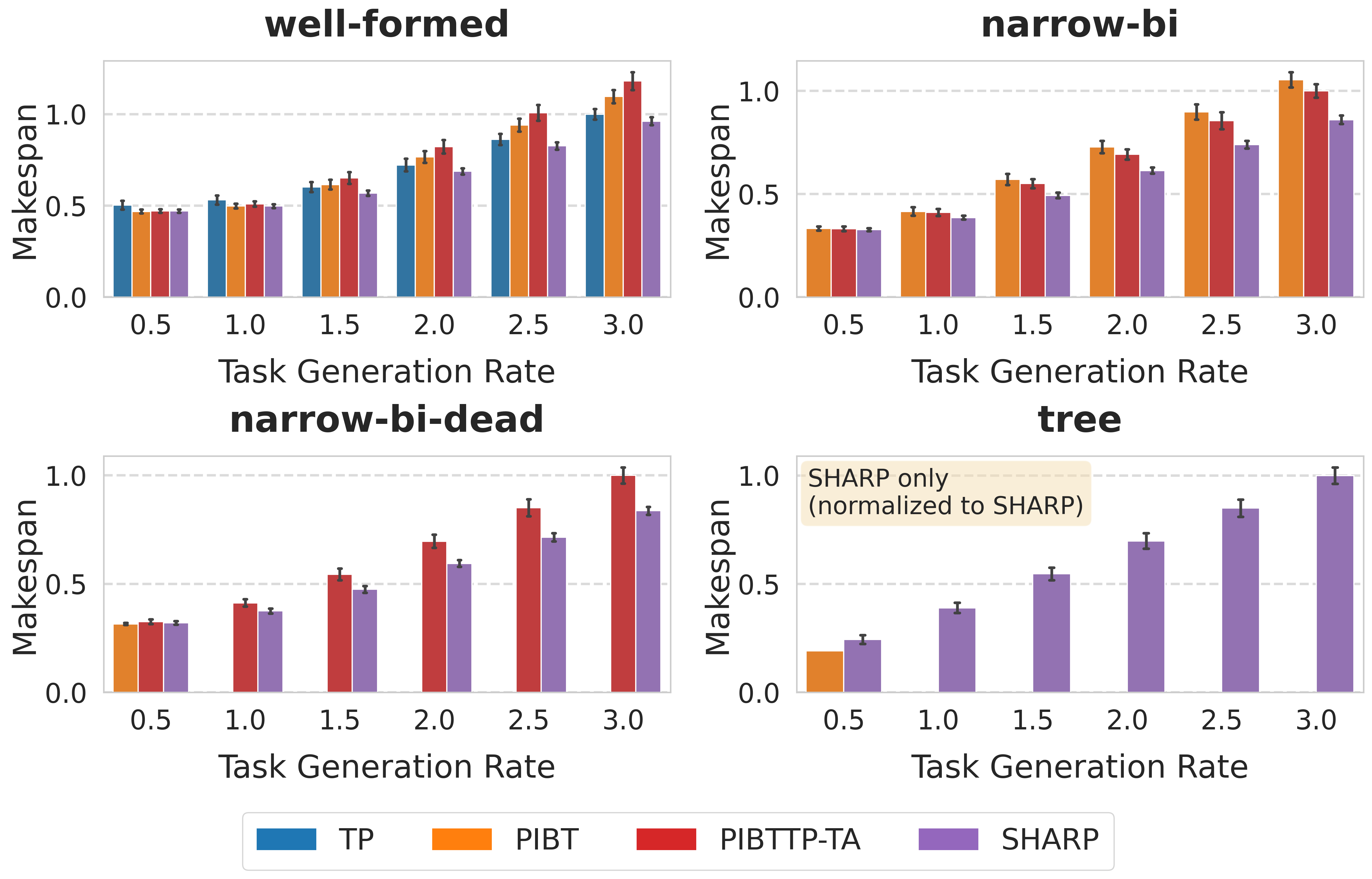}
\caption{Vs. task generation rate (30 agents).}
\label{fig:makespan}
\end{subfigure}
\hfill
\begin{subfigure}[t]{0.48\textwidth}
\centering
\includegraphics[width=\linewidth]{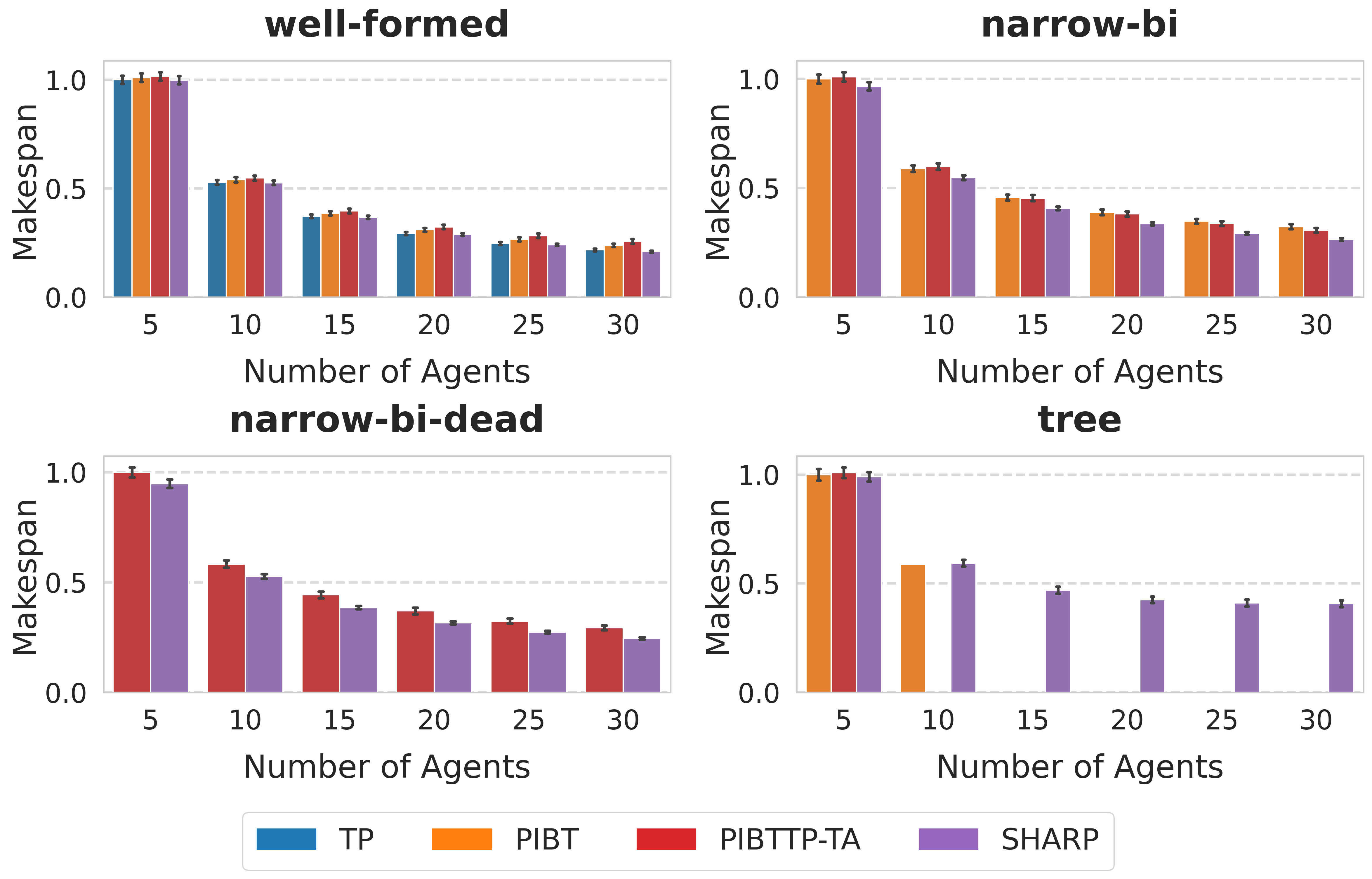}
\caption{Vs. number of agents ($\lambda=3.0$).}
\label{fig:makespan_agents}
\end{subfigure}
\caption{Normalized makespan, with the same plotting conventions as Figure~\ref{fig:service_time_pair}.}
\label{fig:makespan_pair}
\end{figure*}

\subsection{TP-Style Fixed-Home-Return Counterfactual}

Tables~\ref{tab:home_return_counterfactual} and~\ref{tab:home_return_counterfactual_tree} isolate the effect of fixed-home return validation.

\begin{table}[!tbp]
\centering
\caption{TP-style fixed-home-return comparison on Ma et al. well-formed 30-agent conditions (30 paired seeds). Values are absolute means; lower is better except success; runtime is in seconds.}
\label{tab:home_return_counterfactual}
{\scriptsize
\setlength{\tabcolsep}{2.5pt}
\renewcommand{\arraystretch}{1.03}
\begin{tabular}{llrrrr}
\hline
\multicolumn{6}{c}{\textbf{Well-formed (absolute means)}} \\
\hline
$\lambda$ & Method & Succ. & Makespan & Service & Runtime \\
\hline
\multirow{4}{*}{1.0}
& TP & 100.0 & 370.8 & 39.1 & 2.53 \\
& TP-home-return & 100.0 & 537.9 & 120.3 & 2.42 \\
& TP-SIPP-home-return & 100.0 & 538.4 & 120.3 & 2.80 \\
& SHARP & 100.0 & \textbf{349.2} & \textbf{33.0} & 2.86 \\
\hline
\multirow{4}{*}{3.0}
& TP & 100.0 & 702.0 & 180.8 & 4.53 \\
& TP-home-return & 100.0 & 1511.9 & 579.3 & 8.81 \\
& TP-SIPP-home-return & 100.0 & 1509.9 & 579.1 & 10.39 \\
& SHARP & 100.0 & \textbf{673.5} & \textbf{173.9} & 7.56 \\
\hline
\end{tabular}
}
\end{table}

\begin{table}[!tbp]
\centering
\caption{Tree subset of the TP-style fixed-home-return comparison on 30-agent conditions (30 paired seeds). Ratios are normalized to TP-home-return = 1.00; lower is better except success.}
\label{tab:home_return_counterfactual_tree}
{\scriptsize
\setlength{\tabcolsep}{2.5pt}
\renewcommand{\arraystretch}{1.03}
\begin{tabular}{llrrrr}
\hline
\multicolumn{6}{c}{\textbf{Tree (ratios)}} \\
\hline
$\lambda$ & Method & Succ. & Makespan & Service & Runtime \\
\hline
\multirow{4}{*}{1.0}
& TP & 0.0 & -- & -- & -- \\
& TP-home-return & 100.0 & 1.00 & 1.00 & 1.00 \\
& TP-SIPP-home-return & 100.0 & 1.01 & 1.00 & 1.34 \\
& SHARP & 100.0 & \textbf{0.79} & \textbf{0.53} & 20.46 \\
\hline
\multirow{4}{*}{3.0}
& TP & 0.0 & -- & -- & -- \\
& TP-home-return & 100.0 & 1.00 & 1.00 & 1.00 \\
& TP-SIPP-home-return & 100.0 & 1.00 & 1.00 & 1.31 \\
& SHARP & 100.0 & \textbf{0.73} & \textbf{0.54} & 19.04 \\
\hline
\end{tabular}
}
\end{table}

The well-formed absolute values come from the paired 30-seed home-return sweep and therefore differ slightly from the 100-run main-sweep means.
On well-formed 30-agent conditions, all methods achieve 100\% success and SHARP has the best task metrics.
Naive fixed-home return is costly: at $\lambda=3.0$, TP-home-return more than doubles makespan and more than triples service time relative to TP because away-from-home agents must first return before competing.
This penalty is not inherent to reserving a return suffix; it comes from the diagnostic TP semantics that exclude an away-from-home agent from assignment until it has physically reached home.
SHARP keeps the fixed-return safety role but evaluates all idle-or-retreating agents and can overwrite a retreat suffix with a newly validated pickup-delivery-Haven route.
Friedman tests were significant in all four map-rate conditions (max observed $p$-value $< 1.7 \times 10^{-10}$), and all SHARP-versus-home-return comparisons for makespan and service time remained significant after Bonferroni correction (max observed $p$-value $< 1.8 \times 10^{-6}$).
On tree conditions, TP has no successful runs, both fixed-home-return variants recover 100\% success under this TP-style full-route-validation policy, and SHARP further improves normalized makespan/service at substantially higher run-level runtime.

The no-overwrite variant further separates robustness from efficiency.
Table~\ref{tab:no_overwrite_ablation} shows the high-load tree result.

\begin{table}[!tbp]
\centering
\caption{No-overwrite variant on the high-load tree condition (30 agents, $\lambda=3.0$, 30 paired seeds). Ratios are normalized to SHARP = 1.00; lower is better except success.}
\label{tab:no_overwrite_ablation}
{\scriptsize
\setlength{\tabcolsep}{3pt}
\renewcommand{\arraystretch}{1.03}
\begin{tabular}{lcc}
\hline
Metric & SHARP & No-overwrite \\
\hline
Success (\%) & 100 & 100 \\
Makespan & 1.00 & 1.53 \\
Service time & 1.00 & 1.89 \\
Travel distance & 1.00 & 1.53 \\
ms/step & 1.00 & 1.43 \\
Planner calls & 1.00 & 1.24 \\
Expanded nodes & 1.00 & 1.44 \\
Failed validations & 1.00 & 1.27 \\
\hline
\end{tabular}
}
\end{table}

Both variants retain 100\% success, but disabling mid-retreat reassignment worsens all reported task-efficiency and planning-effort ratios; all 30 paired seeds worsen for makespan and service time (one-sided Wilcoxon $p=9.31\times 10^{-10}$).
Because fleet-wide scanning over eligible idle agents is retained, this ablation isolates only the within-SHARP penalty of disabling suffix overwrite; token ordering and non-fleet-wide assignment remain future work.

\subsection{Planning Overhead and Runtime Tradeoff}

\begin{table}[!tbp]
\centering
\caption{Representative SHARP planner effort in the main sweep, reported as mean per successful run. Calls are SIPP segment searches, Expanded is the number of expanded SIPP nodes, and Failed is the number of rejected SIPP segment searches rather than failed simulation runs. Smaller values indicate less search effort.}
\label{tab:search_effort}
{\scriptsize
\setlength{\tabcolsep}{2pt}
\renewcommand{\arraystretch}{1.04}
\begin{tabularx}{\columnwidth}{@{}>{\raggedright\arraybackslash}X r r r@{}}
\hline
Setting & Calls & Expanded & Failed \\
\hline
\shortstack[l]{well-formed\\30 agents, $\lambda=3.0$}
& 4,201 & 0.72M & 0.75 \\
\shortstack[l]{narrow-bi-dead\\30 agents, $\lambda=3.0$}
& 4,249 & 0.84M & 102.19 \\
\shortstack[l]{tree\\30 agents, $\lambda=1.0$}
& 19,530 & 61.08M & 9513.75 \\
\shortstack[l]{tree\\30 agents, $\lambda=3.0$}
& 56,964 & 185.98M & 27790.61 \\
\hline
\end{tabularx}
}
\end{table}

Table~\ref{tab:search_effort} shows the cost hidden by the success heatmap: tree layouts dominate the search overhead.
The fixed-home-return counterfactual shows that this is not only a matter of longer return paths: SHARP obtains better tree task efficiency than the TP-style fixed-home-return baselines, but at substantially higher planning cost.
In deployments, this layer is most plausible when dispatch is slower than low-level control and optimized SIPP kernels, route caches, and reservation-query indexes exist.

\section{Deployment Assumptions and Limitations}
SHARP assumes a centralized fleet manager with a global reservation table and synchronized route execution, so comparisons to decentralized methods such as PIBT speak primarily to \emph{structural applicability}.
The formal guarantee covers fixed exclusive Havens, zero pickup/delivery dwell duration, finitely many task releases, and the fixed-Haven commitment model stated in Section~\ref{sec:framework}.
Task-swap extensions, deterministic nonzero dwell times, stochastic delays, dynamic Haven reassignment, or continuous free-space motion without a stable guidepath graph would require additional mechanisms and proofs.
Reported runtimes are averages, not per-step tail latencies; deployments with hard dispatch deadlines would need tail-latency measurement and optimized planning kernels.
The TP-home-return and no-overwrite diagnostics are not exhaustive, and the tree layouts should be read as tree-like storage/retrieval abstractions rather than copied floor plans; broader decompositions across token ordering, assignment scope, and map classes remain future work.
SHARP is therefore most attractive when natural per-agent home cells or docks already exist and the deployment can afford centralized planning.

\section{Conclusion}

At the framework level, fixed-Haven reservation provides a finite-release-complete coordination framework under Haven-Reachability and explicit planning/progress assumptions.
At the algorithm level, SHARP realizes this idea by maintaining a retreat reservation to a fixed Haven and, against representative TP and PIBT-family MAPD baselines in the main sweep, attains 100\% success on all tested layouts.
The targeted counterfactual shows that fixed-home return with full-route validation under TP-style semantics is sufficient to recover robustness on the representative tree cases tested here, suggesting that the fixed-return contract is a central robustness mechanism in those conditions.
The no-overwrite variant further shows that mid-retreat suffix overwrite is not needed for robustness in the tested high-load tree condition, but it is associated with large service-time and makespan gains while fleet-wide scanning over eligible idle agents is retained.
SHARP's algorithmic contribution is therefore to exploit the fixed-return contract more efficiently, while exposing a substantial centralized planning cost.
Future work will relax the fixed-Haven contract, support richer task-duration models, and weaken support-region conditions for partially used guidepath graphs.

\FloatBarrier
\bibliography{aaai2026}

\end{document}